\documentclass[runningheads,a4paper]{llncs}

\usepackage{amsmath}
\usepackage{amssymb}
\usepackage{wasysym}
\usepackage[T1]{fontenc}
\usepackage{lmodern}
\usepackage[utf8]{inputenc}
\usepackage{subfigure}
\usepackage{pgf}
\usepackage{bussproofs}
\usepackage{tikz}
\usetikzlibrary{arrows,automata,positioning}

\newcommand{\Rec}{\mathop{\mathrm{Rec}}}
\newcommand{\dom}{\mathop{\mathrm{dom}}}
\newcommand{\comp}{\mathop{\mathrm{comp}}}
\newcommand{\vis}{\mathop{\mathrm{vis}}}

\begin{document}
\mainmatter

\title{Joining Transition Systems of Records: Some Congruency and Language-Theoretic Results}
\titlerunning{Joining Transition Systems of Records}
\toctitle{Joining Transition Systems of Records}

\author{Mohammad Izadi\inst{1}
\and Saeed Masoudian\inst{1}
\and Sahand Mozaffari\inst{1}}
\authorrunning{M.~Izadi, S.~Masoudian, S.~Mozaffari}
\tocauthor{M.~Izadi, S.~Masoudian, S.~Mozaffari}

\institute{Department of Computer Engineering, \ \\  Sharif University of Technology, Tehran, Iran.\ \\
 \email{izadi@sharif.edu}, \email{\{masoodian, smozaffari\}@ce.sharif.edu}}

\maketitle

\begin{abstract}
Büchi automaton of records (BAR) has been proposed as a basic operational semantics for Reo coordination language. It is an extension of Büchi automaton by using a set of records as its alphabet or transition labels. Records are used to express the synchrony between the externally visible actions of coordinated components modeled by BARs. The main composition operator on the set of BARs is called as \emph{join} which is the semantics of its counterpart in Reo. In this paper, we define the notion of \emph{labeled transition systems of records} as a generalization of the notion of BAR, abstracting away from acceptance or rejection of strings. Then, we consider four equivalence relations (semantics) over the set of labeled transition systems of records and investigate their congruency with respect to the join composition operator. In fact, we prove that the \emph{finite-traces-based}, \emph{infinite-traces-based}, and \emph{nondeterministic finite automata (NFA)-based} equivalence relations all are congruence relations over the set of all labeled transition systems of records with respect to the join operation. However, the equivalence relation using Büchi acceptance condition is not so. In addition, using these results, we introduce the language-theoretic definitions of the join operation considering both finite and infinite strings notions. Also, we show that there is no language-based and structure-independent definition of the join operation on Büchi automata of records.
\par\addvspace\baselineskip\noindent\keywordname\enspace\ignorespaces{Transition Systems of Records, Congruence Relations, Equivalence Relations, Semantics, Büchi Automata of Records, Language Theory, Reo Coordination Language, Join Operation}
\end{abstract}

\section{Introduction}
In a series of papers, Büchi Automaton of Records (BAR) has been proposed as a basic semantics formalism for coordination of computing systems~\cite{IB08,IBC08,Izadi-SoSym}. It has been shown that BAR not only is a suitable modeling formalism to be used as the semantics of all primitive and user-defined coordination strategies specified by the predefined channels and their composed nets in the language of Reo~\cite{Arb02}, but also it is basically more powerful to express context dependencies and fairness constraints in comparison with the first proposed operational semantics of Reo that is called as Constraint Automaton (for Constraint Automaton see~\cite{CA06} and for BAR and its expressiveness and comparisons see~\cite{Izadi-SoSym,Izadi-PhD}). 

A Büchi Automaton of Records is a standard Büchi Automaton in which each transition label (alphabet member) is a \emph{record} $r$ over a set of port names $\mathcal{N}$  and data set $\mathcal{D}$. A record, formally, is a partial function from  $\mathcal{N}$  to $ \mathcal{D}$ and, semantically, is a data structure for expressing the simultaneous data communication events over the set of ports: the ports that are in the record’s domain are able to communicate their assigned data simultaneously, while the other ports are blocked and not allowed to communicate.  Using Büchi acceptance condition, a BAR is considered as the acceptor of infinite strings (streams) of records.

To obtain complex coordination systems by composing simpler ones modeled by BARs, we can use the \emph{standard product} of Büchi automata which can be applied whenever the sets of the alphabets (transition labels) are the same and it will be the counterpart of the intersection of the automata languages. Moreover, using the richer structure of BAR’s alphabet, we can define a more general product operator that works if the alphabets of the automata are different. We call this as the \emph{join} operator which is the semantics counterpart of the syntactic join operator in Reo. In the case of BARs over the same set of records, the join and standard product operators are structurally and language theoretically the same. But in the case of BARs over different set of records (in fact, different but not necessarily disjoint sets of port names) the join operation models joining different coordinators using only their common ports. In the special case of BARs over disjointed alphabets, the join operator models the asynchronous product (interleaving the executions) of the first automata.

Obviously, abstracting away from the notion of Büchi acceptance condition and saving the idea of using records as the transition labels of a labeled transition system enables us to consider several other useful theoretical results and even problems in the field of formal semantics and their counterparts in the field of specification languages. Let us call this type of transition systems as \emph{labeled transition systems of records} (LTSR) and their composition operator using the rules of the joining of BARs again as the \emph{join} operator. Now, we are faced with a lot of useful semantic models and equivalence relations such as \emph{finite-traces-based}, \emph{finite automata (NFA-based)},  \emph{infinite-traces-based},  \emph{Büchi, Robin, Muller,… acceptance conditions based}, \emph{failure based} , and \emph{(bi)simulation} equivalence relations over the set of all labeled transition systems of records (for a survey on a large set  of these kinds of equivalence relations see~\cite{van-glabbeek-spectrum-I,van-glabbeek-spectrum-II}).  

Whenever an equivalence relation over a set of formal models is considered as their semantics and the models can be composed using a predefined composition operator, one of the main questions is the \emph{congruency} of the equivalence relation with respect to the composition operator. In fact, an equivalence relation is a congruent relation with respect to a composition operator if the replacement of a component of a composed model by an equivalent one should always yield a model which is equivalent with the original one.  This congruency satisfaction guarantees for example the applicability and practical usefulness of the selected equivalence relation and composition operator in compositional reduction (minimization) of models in the process of verification and model checking.

In this paper, we investigate the congruency of four equivalence relations over the set of labeled transition systems of records (LTRS) with respect to their join composition operator. The equivalence relations are:  \emph{finite-traces-based} relation by which two LTSRs are equivalent if their sets of finite traces are the same, \emph{infinite-traces-based} relation by which two LTSRs without any deadlock state are equivalent if their sets of infinite traces are the same, \emph{finite automata (NFA-based)} relation over the set of LTRSs with subsets of states as the accepting sets in which two LTSRs are equivalent if their sets of accepted finite traces are the same, and \emph{ Büchi automata based} relation over the set of LTRSs with subsets of accepting states in which two LTSRs are equivalent if their sets of infinite traces that satisfy the Büchi acceptance condition are the same. In fact, we show that the first three equivalence relations all save the congruency with respect to the join operator , however the last one does not so. In addition, we are interested to define the join operation over LTSRs as language theoretic operations respectively over the set of finite or infinite strings of records in each of the above mentioned semantic definitions (equivalence relations). We introduce the language theoretic definitions for the join operation for the first three cases and prove that for the case of Büchi automata of records there is no such language theoretic and structure independent definition. 

The paper proceeds as follows: in section 2 we introduce our used basic notations and definitions and review the required preliminaries. Section 3 dedicated to the main results, theorems and corollaries. In section 4, we present our conclusions and suggest some future works.

\section{Preliminaries}
In this section we introduce the basic notations and preliminaries that we use in the rest of this paper. Firstly, we shall define the notion of record, which constitutes main objects of this paper.

\begin{definition}
Let $\mathcal{N}$ be a set of port names and $\mathcal{D}$ be a set of data:
\begin{enumerate}
\item $\emph{Records}$ are partial functions $r: \mathcal{N} \rightarrow \mathcal{D}$. Here we call $\mathcal{N}$, \emph{name set} or \emph{label set} of $r$, and $\mathcal{D}$ is called its \emph{data set}.
\item Denoted by $\Rec_\mathcal{N}(\mathcal{D})$ is the set of all records with labels from $\mathcal{N}$ and data from $\mathcal{D}$.
\item For a record $r \in \Rec_\mathcal{N}(\mathcal{D})$, we write $\dom(r)$ for domain of $r$.
\item The special record with empty domain is called \emph{invisible record} and is denoted by $\tau$.
\item A \emph{stream of records} over name set $\mathcal{N}$ and data set $\mathcal{D}$ is a (possibly infinite) string of records from $\Rec_\mathcal{N}(\mathcal{D})$. Denoted by $\Rec_\mathcal{N}(\mathcal{D})^*$ and $\Rec_\mathcal{N}(\mathcal{D})^\omega$ are the sets of all finite and infinite streams of records of $\Rec_\mathcal{N}(\mathcal{D})$, respectively.
\end{enumerate}
\end{definition}

Records are used to convey both positive and negative information. That is, only the ports in $\dom(r)$ can possibly exchange the data assigned to them, while the other ports can't perform any communication. A record's information, can be altered through the following operations.
\begin{definition}
Let $r \in \Rec_{\mathcal{N}}(\mathcal{D})$:
\begin{enumerate}
\item For $\mathcal{N'} \subseteq \mathcal{N}$, $r\downarrow_\mathcal{N'}$ is $r$'s restriction to $\mathcal{N'}$. This operation increases a record's negative information.
\item Restriction is defined similarly for strings as $w\downarrow_\mathcal{N'} = w'$, where $w'[i]$ is $w[i]\downarrow_\mathcal{N'}$.
\item \emph{visible portion} of a string $w$ is defined as the sequence of its visible symbols and is denoted by $\vis(w)$.
\end{enumerate}
\end{definition}

When two records are not in contradiction with respect to the positive information they hold, it may be desirable to aggregate their information through \emph{union} operation as defined below.
\begin{definition}
Let $r_1 \in \Rec_{\mathcal{N}_1}(\mathcal{D})$ and $r_2 \in \Rec_{\mathcal{N}_2}(\mathcal{D})$. Then:
\begin{enumerate}
\item We say $r_1$ and $r_2$ are \emph{compatible} records and denote it by $\comp(r_1, r_2)$ when:
\[
\dom(r_1) \cap \mathcal{N}_2 = \dom(r_2) \cap \mathcal{N}_1 \wedge \forall \left. n \in \dom(r_1) \cap \dom(r_2)\right. : r_1(n) = r_2(n)
\]

\item The \emph{union} of compatible records $r_1$ and $r_2$, denoted by $r_1 \cup r_2$, is a record over name set $\mathcal{N}_1 \cup \mathcal{N}_2$, defined as:
\[
(r_1 \cup r_2)(n) = \begin{cases}
r_1(n)	&	n \in \dom(r_1)	\\
r_2(n)	&	n \in \dom(r_2)	\\
\end{cases}
\]
\end{enumerate}
\end{definition}

It is often convenient to model a system with a graph-like entity called \emph{labeled transition system}. We now give the description of this abstract machine.

\begin{definition}
\begin{enumerate}
\item A \emph{labeled transition system} is a quadruple $\mathcal{M} = (Q, \Sigma, \Delta, Q_0)$, where $Q$ is a finite nonempty set of \emph{states}, $\Sigma$ is a finite nonempty set of symbols called \emph{alphabet}, $\Delta: Q \times \Sigma \rightarrow \mathcal{P}(Q)$ is \emph{transition function} and $Q_0 \subseteq Q$ is the nonempty set of \emph{initial states}. In the rest of this paper, we will use the abbreviation LTS.
\item A string $w$ (possibly infinite) on alphabet $\Sigma$ is said to be \emph{traceable} in $\mathcal{M}$ when there is a sequence of \emph{states}, $q_i \in Q$, called \emph{trace}, where $q_0 \in Q_0$ and $q_{i+1} \in \Delta(q_i, w[i])$.
\end{enumerate}
\end{definition}
Of special interest to us, are those LTSs whose alphabets are $\Sigma = \Rec_\mathcal{N}(\mathcal{D})$ for some name set $\mathcal{N}$ and data set $\mathcal{D}$. We refer to these as \emph{LTSs of streams of records}, or LTSR in short.

To capture the notion of liveness in a system, we define a more general model, called Büchi automaton.
\begin{definition}
\begin{enumerate}
\item A \emph{Büchi automaton} is a quintuple $\mathcal{M} = (Q, \Sigma, \Delta, Q_0, F)$, where $Q$ is a finite nonempty set of \emph{states}, $\Sigma$ is a finite nonempty set called \emph{alphabet}, a \emph{transition function} $\Delta: Q \times \Sigma \rightarrow \mathcal{P}(Q)$, the nonempty set of \emph{initial states} $Q_0 \subseteq Q$ and the nonempty set of \emph{final states} $F \subseteq Q$.
\item For finite string $w$ on alphabet $\Sigma$ of length $n$, we say $\mathcal{B}: Q_1 \xrightarrow{w} Q_2$ when $Q_2$ is the set of states $q_n$, such that there exists a finite sequence of states $q_i \in Q$, where $q_0 \in Q_1$ and $q_n \in Q_2$ and $q_{i+1} \in \Delta(q_i, w_i)$.
\item An infinite string $w$ on alphabet $\Sigma$ is said to be accepted in $\mathcal{M}$, when there exists a trace of $w$ in $\mathcal{B}$ which infinitely often visits final states.
\item The set of all accepted infinite strings in a Büchi automaton, denoted by $L_B(\mathcal{M})$, is called its \emph{language}.
\item $L_f(\mathcal{B})$ is the set of all finite strings $w$ on alphabet $\Sigma$, such that $\mathcal{B}: Q_0 \xrightarrow{w} Q'$ and $Q' \cap F \neq \emptyset$.
\end{enumerate}
\end{definition}
Basically, a Büchi automaton is an LTS equipped with a set of differentiated states, called final states, such that the whole system is considered to be live when it keeps visiting the set of final states. $\xrightarrow{w}$ captures the notion of reachability. $L_B$ is the set of infinite strings that keep the system live and $L_f$ is the language of system when regarded as an NFA.

Here again, we especially interested in those of Büchi automata defined on an alphabet of records. We call these \emph{Büchi automata on stream of records}, or BAR in short.

\begin{proposition}
\label{prop:lts-as-automaton}
A (possibly infinite) string $w$ is traceable in an LTS defined as $\mathcal{M} = (Q, \Sigma, \Delta, Q_0)$, if and only if it is in the language of the automaton $\mathcal{A} = (Q, \Sigma, \Delta, Q_0, Q)$.
\end{proposition}

An equivalent model is the model of \emph{generalized Büchi automaton}, which sometimes gives a more convenient way to use.
\begin{definition}
\begin{enumerate}
\item A \emph{generalized Büchi automaton} is defined as a quintuple $\mathcal{M} = (Q, \Sigma, \Delta, Q_0, \mathcal{F})$, where $Q$ is a finite nonempty set of \emph{states}, $\Sigma$ is a finite nonempty set called \emph{alphabet}, a \emph{transition function} $\Delta: Q \times \Sigma \rightarrow \mathcal{P}(Q)$, the nonempty set of \emph{initial states} $Q_0 \subseteq Q$ and the family of sets of \emph{final states} $\mathcal{F} \subseteq \mathcal{P}(Q)$.
\item An infinite string $w$ on alphabet $\Sigma$ is said to be \emph{accepted} in $\mathcal{M}$, when $w$ has a trace in $\mathcal{B}$ which for each $F \in \mathcal{F}$ visits $F$ infinitely often.
\item Similarly, $L_f(\mathcal{B})$ is defined as the set of all finite strings $w$ on alphabet $\Sigma$, such that $Q' \cap \bigcap_{F \in \mathcal{F}} F \neq \emptyset$, where $\mathcal{B}: Q_0 \xrightarrow{w} Q'$.
\end{enumerate}
\end{definition}
\begin{proposition}
For every generalized Büchi automaton $B$, there exists a Büchi automaton $B'$, where $L_B(B) = L_B(B')$. \cite{Vardi automata approach to LTL,Thomas90}
\end{proposition}

Finally, we will define the join operator, with which this paper is mainly concerned.
\begin{definition}
Let
$\mathcal{B}_i = (Q_i, \Rec_{\mathcal{N}_i}(\mathcal{D}), \Delta_i, Q_{0i}, F_i)$
for $i=1,2$ be two BARs. Denoting by
$\mathcal{B}_1 \Bowtie \mathcal{B}_2$,
we define the join of $\mathcal{B}_1$ and $\mathcal{B}_2$ as the generalized Büchi automaton:
\[
\mathcal{B}_1 \Bowtie \mathcal{B}_2 = (Q_1 \times Q_2, {\textstyle \Rec_{\mathcal{N}_1 \cup \mathcal{N}_2}(\mathcal{D})}, \Delta, Q_{01} \times Q_{02}, \{ F_1 \times Q_2, Q_1 \times F_2 \})
\]
where $\Delta$ is given by the following rules:
\begin{prooftree}
\AxiomC{$\Delta_1(q_1, r_1)=q'_1$}
\AxiomC{$\Delta_2(q_2, r_2) = q'_2$}
\AxiomC{$\comp(r_1, r_2)$}
\TrinaryInfC{$\Delta((q_1, q_2), r_1 \cup r_2) = (q'_1, q'_2)$}
\end{prooftree}
\begin{prooftree}
\AxiomC{$\Delta_1(q_1, r_1)=q'_1$}
\AxiomC{$\dom(r_1) \cap \mathcal{N}_2 = \emptyset$}
\BinaryInfC{$\Delta((q_1, q_2), r_1 \cup r_2) = (q'_1, q'_2)$}
\end{prooftree}
\begin{prooftree}
\AxiomC{$\Delta_1(q_2, r_2)=q'_2$}
\AxiomC{$\dom(r_2) \cap \mathcal{N}_1 = \emptyset$}
\BinaryInfC{$\Delta((q_1, q_2), r_1 \cup r_2) = (q'_1, q'_2)$}
\end{prooftree}

The less common but more convenient way to define $\Delta$ is the following.
\[
\Delta\left( (q_1, q_2), r \right) = \begin{cases}
\left( \Delta(q_1, r\downarrow_{\mathcal{N}_1}), q_2 \right)	&	\mathcal{N}_2 \cap \dom(r) = \emptyset	\\
\left( q_1, \Delta(q_2, r\downarrow_{\mathcal{N}_2}) \right)	&	\mathcal{N}_1 \cap \dom(r) = \emptyset	\\
\left( \Delta(q_1, r\downarrow_{\mathcal{N}_1}), \Delta(q_2, r\downarrow_{\mathcal{N}_2}) \right)	&	\text{otherwise}
\end{cases}
\]
\end{definition}
We will make use of the latter equivalent definition of $\Delta$ in this text.

Proposition~\ref{prop:lts-as-automaton}, states an equivalence between LTSRs and BARs. By this equivalence, join operation can be defined similarly for LTSRs.

The rest of this paper is concerned with the congruency of common equivalences on LTSRs adn BARs with respect to join operation. Here is what we mean by congruency:
\begin{definition}
Let $\odot$ be a $k$-ary operator defined on a set $S$. We say $\approx$ is a \emph{congruence relation} on $S$ with respect to $\odot$ if it is an equivalence relation and $\odot(a_1, \ldots, a_k) \approx \odot(b_1, \ldots, b_k)$ whenever $a_i \approx b_i$ for $1 \leq i \leq k$.
\end{definition}

\begin{proposition}
\label{prop:congruency_binary_commutative}
Let $\odot$ be a binary commutative operator over a set $S$. An equivalence relation $\approx$ is a congruency on $S$ with respect to $\odot$, if and only if $a_1 \odot b \approx a_2 \odot b$ whenever $a_1 \approx a_2$.
\end{proposition}

\section{Main Results}
In this section we present our main results, which are the proofs of congruency of some equivalence relations on LTSRs and BARs with respect to join operation. In addition, we use these results to give  language-based definitions of join operation in the cases where the congruency property holds.
Before we enter the main subject, we introduce some basic lemmas.
\begin{lemma}
\label{lem:join-state}
Let
$\mathcal{M}_i = (Q_i, \Rec_{\mathcal{N}_i}(\mathcal{D}), \Delta_i, Q_{0i})$
be LTSRs for $i = 1, 2$. Let $w$ be a (possibly infinite) string over alphabet $\Rec_{\mathcal{N}_1 \cup \mathcal{N}_2}(\mathcal{D})$ and define $w' = w\downarrow_{\mathcal{N}_1}$ and $w'' = w\downarrow_{\mathcal{N}_2}$. For any $j$,
$\mathcal{M}_1 \Bowtie \mathcal{M}_2: Q_{01} \times Q_{02} \xrightarrow{w[1..j]} Q'_1 \times Q'_2$
if and only if
$\mathcal{M}_1: Q_{01} \xrightarrow{\vis(w'[1..j])} Q'_1$
and
$\mathcal{M}_2: Q_{02} \xrightarrow{\vis(w''[1..j])} Q'_2$.
\end{lemma}

\begin{proof}
First, we demostrate that if
$\mathcal{M}_1 \Bowtie \mathcal{M}_2: Q_{01} \times Q_{02} \xrightarrow{w[1..j]} Q'_1 \times Q'_2$
then
$\mathcal{M}_1: Q_{01} \xrightarrow{\vis(w'[1..j])} Q'_1$.
The statement about $\mathcal{M}_2$ follows from a similar argument.
We employ an induction on $j$. Induction basis for $j=1$ follows immediately from definition. Assume that the proposition holds for $i < k$. We aim to prove the case where $j=k$. We know in particular that:
\[
\mathcal{B}_1 \Bowtie \mathcal{B}_2: Q_{01} \times Q_{02} \xrightarrow{w[1..(k-1)]} Q'_1 \times Q'_2
\Rightarrow
\mathcal{B}_1: Q_{01} \xrightarrow{\vis(w'[1..(k-1)])} Q'_1
\]

Define $Q''_1$ and $Q''_2$ so that
$\mathcal{B}_1 \Bowtie \mathcal{B}_2: Q_{01} \times Q_{02} \xrightarrow{w[1..k]} Q''_1 \times Q''_2$.
If $w'[k] = \tau$, then $Q'_1 = Q''_1$ and the proposition is self-evident. Thus, assume that $w'[k] \neq \tau$. It remains to prove that:
\[
\mathcal{B}_1 \Bowtie \mathcal{B}_2: Q'_1 \times Q'_2 \xrightarrow{w[k]} Q''_1 \times Q''_2
\Rightarrow
\mathcal{B}_1: Q'_1 \xrightarrow{w'[k]} Q''_1
\]

Let $r_1 = w[k]\downarrow_{\mathcal{N}_1}$.
By assumption $w'[k] \neq \tau$. Thus $\dom(r) \cap \mathcal{N}_1 \neq \emptyset$. Hence $w'[k] = r\downarrow_{\mathcal{N}_1}$ and $Q''_1 = \Delta_{\mathcal{B}_1}(Q'_1, r\downarrow_{\mathcal{N}_1})$, by definition. Therefore the proposition holds.

To prove the converse statement, we again use induction on $j$. Induction basis is obvious for $j=1$. Suppose that the claim is valid for all $i<k$. Particularly:
\[
\left.\begin{gathered}
\mathcal{B}_1: Q_{01} \xrightarrow{\vis(w'[1..(k-1)])} Q'_1	\\
\mathcal{B}_2: Q_{02} \xrightarrow{\vis(w''[1..(k-1)])} Q'_2
\end{gathered}\right\}
\Rightarrow
\mathcal{B}_1 \Bowtie \mathcal{B}_2: Q_{01} \times Q_{02} \xrightarrow{w[1..(k-1)]} Q'_1 \times Q'_2
\]
Now define $Q''_1$ and $Q''_2$ so that $Q'_1 \xrightarrow{\vis(w'[k])} Q''_1$ and $Q'_2 \xrightarrow{\vis(w''[k])} Q''_2$. It suffices to prove that
$\mathcal{B}_1 \Bowtie \mathcal{B}_2: Q'_1 \times Q'_2 \xrightarrow{w[k]} Q''_1 \times Q''_2$.
If $w'[k] \cap \mathcal{N}_1 = \emptyset$, then $\vis(w'[k])$ is the empty string and $Q'_1 = Q''_1$. Otherwise $Q''_1 = \Delta_1(Q_1, \vis(w'[k]))$. A similar statement holds for $Q''_2$. It is easily observed from definition of $\Delta_{\mathcal{B}_1 \Bowtie \mathcal{B}_2}$ that
$\mathcal{B}_1 \Bowtie \mathcal{B}_2: Q'_1 \times Q'_2 \xrightarrow{w[k]} Q''_1 \times Q''_2$, in all cases. Thus the proof is complete.
\end{proof}

\begin{lemma}
\label{lem:main-lemma}
Let
$\mathcal{B}_i = (Q_i, \Rec_{\mathcal{N}_i}(\mathcal{D}), \Delta_i, Q_{0i})$
be BARs for $i = 1, 2, 3$. Then
$L_B(\mathcal{B}_2 \Bowtie \mathcal{B}_3) = L_B(\mathcal{B}_1 \Bowtie \mathcal{B}_3)$
for any arbitrary $\mathcal{B}_3$, if and only if
$L_B(\mathcal{B}_1)  = L_B(\mathcal{B}_2)$ and $L_f(\mathcal{B}_1) = L_f(\mathcal{B}_2)$.
\end{lemma}

\begin{proof}
\textbf{Necessity of $L_f(\mathcal{B}_1) = L_f(\mathcal{B}_2)$:} First, we show that it is necessary for $\mathcal{B}_1$ and $\mathcal{B}_2$ to have the same language as NFAs. Assume that for some $\mathcal{B}_1$ and $\mathcal{B}_2$,
$L_f(\mathcal{B}_1) \neq L_f(\mathcal{B}_2)$.
we show that there exists a Büchi automaton $\mathcal{B}_3$ such that
$\mathcal{B}_1 \Bowtie \mathcal{B}_3 \neq \mathcal{B}_2 \Bowtie \mathcal{B}_3$.
Define
$\mathcal{B}_3 = (\{q_{03}\}, \{f\}, \Delta_3, \{q_{03}\}, \{q_{03}\})$,
where
$f$ is a record whose domain is disjoint from $\mathcal{N}_1$ and $\mathcal{N}_2$
and let
$\Delta_3(q_{03}, f) = \{q_{03}\}$.
Figure~\ref{fig:language_congruency_necessary_condition} depicts this automaton.

\begin{figure}
\centering
\begin{tikzpicture}
\node[initial,state,accepting] (A) {$q_0$};
\path (A) edge[loop right] node{$f$} (A);
\end{tikzpicture}
\caption{Automaton $\mathcal{B}_3$ in proof of lemma~\ref{lem:join-state}}
\label{fig:language_congruency_necessary_condition}
\end{figure}
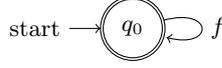

Without loss of generality, assume that there exists a finite word $w \in L_f(\mathcal{B}_1)$ such that $w \not\in L_f(\mathcal{B}_2)$.
Define $w' = wf^\omega$. If
$\mathcal{B}_1: Q_{01} \xrightarrow{w} Q'_1$
and
$\mathcal{B}_2: Q_{02} \xrightarrow{w} Q'_2$,
then $Q'_1 \cap F_1 \neq \emptyset$ and $Q'_2 \cap F_2 = \emptyset$. Using lemma~\ref{lem:join-state}, we derive:
\[
\mathcal{B}_1 \Bowtie \mathcal{B}_3: Q_{01} \times \{q_{03}\} \xrightarrow{w} Q'_1 \times \{q_{03}\}
\wedge
\mathcal{B}_2 \Bowtie \mathcal{B}_3: Q_{02} \times \{q_{03}\} \xrightarrow{w} Q'_2 \times \{q_{03}\}
\]
On the other hand, we know that $\Delta_1(q, f) = \{q\}$ for all $q \in Q_1$ and that $\Delta_2(q, f) = \{q\}$ for all $q \in Q_2$.
Thus the string $wf^\omega$ will make $\mathcal{B}_1$ make infinite loops on $Q'_1 \times \{q_{03}\}$ and will make $\mathcal{B}_2$ make infinite loops on the $Q'_2 \times \{q_{03}\}$. Thus
$wf^\omega \in L_B(\mathcal{B}_1 \Bowtie \mathcal{B}_3)$
and
$wf^\omega \not\in L_B(\mathcal{B}_2 \Bowtie \mathcal{B}_3)$.
Therefore
$L_B(\mathcal{B}_1 \Bowtie \mathcal{B}_3)$
and
$L_B(\mathcal{B}_2 \Bowtie \mathcal{B}_3)$
cannot be the same.

\textbf{Necessity of $L_B(\mathcal{B}_1)=L_B(\mathcal{B}_2)$:} Here we prove that $L_B(\mathcal{B}_1)=L_B(\mathcal{B}_2)$ is a necessary condition as well. Without loss of generality, assume that there is $w \in \mathcal{B}_1$ such that $w \not \in L_B(\mathcal{B}_2)$. Define $\mathcal{B}_3$ as before. One can easily see in the same way that $w \in \mathcal{B}_1 \Bowtie \mathcal{B}_3$ but $w \not\in \mathcal{B}_2 \Bowtie \mathcal{B}_3$. This necessitates $L_B(\mathcal{B}_1)=L_B(\mathcal{B}_2)$ for
$L_B(\mathcal{B}_1 \Bowtie \mathcal{B}_3) = L_B(\mathcal{B}_2 \Bowtie \mathcal{B}_3)$
to hold.

\textbf{Sufficiency of the aforementioned conditions:} Finally we demonstrate that the two conditions are also sufficient. Suppose $L_B(\mathcal{B}_1)=L_B(\mathcal{B}_2)$ and $L_f(\mathcal{B}_1)=L_f(\mathcal{B}_2)$. If there is any path from $Q_{01}$ to any final state of $\mathcal{B}_1$, which contains a record $r \not\in \Rec_{\mathcal{N}_2}(\mathcal{D})$, then there is a finite word $w \in L_f(\mathcal{B}_1)$ such that $w \not\in L_f(\mathcal{B}_2)$, contradicting the assumption. Hence, we may assume that $\Rec_{\mathcal{N}_1}(\mathcal{D}) \subseteq \Rec_{\mathcal{N}_2}(\mathcal{D})$. For the same reason, we assume $\Rec_{\mathcal{N}_2}(\mathcal{D}) \subseteq \Rec_{\mathcal{N}_1}(\mathcal{D})$ and conclude that $\mathcal{N}_1 = \mathcal{N}_2$.

We claim that
$L_B(\mathcal{B}_1 \Bowtie \mathcal{B}_3) \subseteq L_B(\mathcal{B}_2 \Bowtie \mathcal{B}_3)$.
Take an arbitrary string $w$ from
$L_B(\mathcal{B}_1 \Bowtie \mathcal{B}_3)$,
and define strings
$w' = w\downarrow_{\mathcal{N}_1}$
and
$w''' = w\downarrow_{\mathcal{N}_3}$.
By definition, automaton $\mathcal{B}_1 \Bowtie \mathcal{B}_3$ visits a state $(f_1, q_3)$, infinitely many times during its operation on $w$, where $f_1 \in F_1$ and $q_3 \in Q_3$. According to lemma~\ref{lem:join-state}, $\mathcal{B}_1$ lands on $f_1$ during its operation on $\vis(w'[1..i])$ for infinitely many $i$. There are two cases for $w'$ to consider:
\begin{enumerate}
\item If $\vis(w')$ is finite, then there is an index $i$, such that $w'[i..\infty] = \tau^\omega$. In this case, $\mathcal{B}_1: Q_{01} \xrightarrow{\vis(w'[1..j])} Q'_1$ for all $j>i$, for some $Q'_1 \subseteq Q_1$. Since $Q'_1 \times Q_3$ is the only subset of $Q_1 \times Q_3$ visited infinitely many times by $\mathcal{B}_1 \Bowtie \mathcal{B}_3$, therefore $Q'_1 \cap F_1 \neq \emptyset$. Thus $\vis(w') \in L_f(\mathcal{B}_1) = L_f(\mathcal{B}_2)$. Hence, $\mathcal{B}_2$ will halt on a final state $f_2 \in F_2$ by its operation on $\vis(w')$. Thus
$\mathcal{B}_2: Q_{02} \xrightarrow{\vis(w'[1..j])} Q'_2$, where $Q'_2 \cap F_2 \neq \emptyset$.
As shown in lemma~\ref{lem:join-state}, if $\mathcal{B}_1 \Bowtie \mathcal{B}_3 : Q_{01} \times Q_{03} \xrightarrow{w[1..k]} Q'_1 \times Q'_3$, then $\mathcal{B}_3: Q_{03} \xrightarrow{\vis(w'''[1..k])} Q'_3$. Thus for every integer $k>i$, we have
$\mathcal{B}_2 \Bowtie \mathcal{B}_3 : Q_{02} \times Q_{03} \xrightarrow{w[1..k]} Q'_2 \times Q'_{j3}$
whenever
$\mathcal{B}_1 \Bowtie \mathcal{B}_3 : Q_{01} \times Q_{03} \xrightarrow{w[1..k]} Q'_1 \times Q'_{j3}$.
Since $Q'_2 \cap F_2 \neq \emptyset$, and $Q'_{j3} \cap F_3 \neq \emptyset$ for infinitely many $j$, then $w \in L_B(\mathcal{B}_2 \Bowtie \mathcal{B}_3)$.

\item Assume $\vis(w')$ is infinite. We know that there are infinitely many $i$ such that
$\mathcal{B}_1 \Bowtie \mathcal{B}_3: Q_{01} \times Q_{03} \xrightarrow{w[1..i]} Q'_{i1} \times Q'_{i3}$,
and $Q'_{i1} \cap F_1 \neq \emptyset$. Using this and lemma~\ref{lem:join-state}, we derive
$\vis(w') \in L_B(\mathcal{B}_1) = L_B(\mathcal{B}_2)$.
Therefore for infinitely many $j$,
$\mathcal{B}_2: Q_{02} \xrightarrow{\vis(w'[1..j])} Q'_{j2}$,
where $Q'_{j2} \cap F_2 \neq \emptyset$.
We also know that there are infinitely many $k$ such that
$\mathcal{B}_1 \Bowtie \mathcal{B}_3: Q_{01} \times Q_{03} \xrightarrow{w[1..k]} Q'_{k1} \times Q'_{k3}$
and $Q'_{k3} \cap F_3 \neq \emptyset$.
Thus, by lemma~\ref{lem:join-state}, there are infinitely many $k$ such that
$\mathcal{B}_3 : Q_{03} \xrightarrow{\vis(w'''[1..k])} Q'_{k3}$,
where $Q'_{k3} \cap F_3 \neq \emptyset$.
For the last step of this proof, lemma~\ref{lem:join-state} indicates that there are infinitely many indices $j$, such that
$\mathcal{B}_2 \Bowtie \mathcal{B}_3: Q_{02} \times Q_{03} \xrightarrow{w[1..j]} Q'_{j2} \times Q'_{j3}$,
where $Q'_{j2} \cap F_2 \neq \emptyset$, and there are infinitely many integers $k$ such that
$\mathcal{B}_2 \Bowtie \mathcal{B}_3: Q_{02} \times Q_{03} \xrightarrow{w[1..k]} Q'_{k2} \times Q'_{k3}$,
where $Q'_{k3} \cap F_3 \neq \emptyset$. Thus $w \in L_B(\mathcal{B}_2 \Bowtie \mathcal{B}_3)$.
\end{enumerate}
This proves that
$L_B(\mathcal{B}_1 \Bowtie \mathcal{B}_3) \subseteq L_B(\mathcal{B}_2 \Bowtie \mathcal{B}_3)$.
A similar argument reveals that
$L_B(\mathcal{B}_2 \Bowtie \mathcal{B}_3) \subseteq L_B(\mathcal{B}_1 \Bowtie \mathcal{B}_3)$.
Together they yield
$L_B(\mathcal{B}_1 \Bowtie \mathcal{B}_3) = L_B(\mathcal{B}_2 \Bowtie \mathcal{B}_3)$,
thus proving the sufficiency of the aforementioned conditions.
\end{proof}

Let
$\mathcal{M}_i = (Q_i, \Rec_{\mathcal{N}_i}(\mathcal{D}),\Delta_i, Q_{0i})$
be LTSRs, for $i=1,2$. Finite-traces-based equivalence, denoted by $\mathcal{M}_1 \approx_{ft} \mathcal{M}_2$, holds whenever for any finite word $w$ on alphabet $\Rec_{\mathcal{N}_i}(\mathcal{D})$, $w$ is traceable in $\mathcal{M}_1$ if and only if it is traceable in $\mathcal{M}_2$.
In the following theorem we show that the finite-traces-based equivalence relation on the set of LTSRs is a congruency with  respect to join operation:
\begin{theorem}
The equivalence relation $\approx_{ft}$ is a congruence relation over the set of all LTSRs with respect to join operation.
\end{theorem}

\begin{proof}
Let
$\mathcal{M}_i = (Q_i, \Rec_{\mathcal{N}_i}(\mathcal{D}), \Delta_i, Q_{0i})$
be LTSRs, and define the equivalent Büchi automata
$\mathcal{B}_i = (Q_i, \Rec_{\mathcal{N}_i}(\mathcal{D}),\Delta_i, Q_{0i}, Q_i)$,
for $i=1,2,3$. A finite string $w$ is traceable in $\mathcal{M}_i$ if and only if $w \in L_f(\mathcal{B}_i)$. Suppose $\mathcal{M}_1 \approx_{ft} \mathcal{M}_2$. It is safe to assume that $\mathcal{N}_1 = \mathcal{N}_2$. We show that
$\mathcal{M}_1 \Bowtie \mathcal{M}_3 \approx_{ft} \mathcal{M}_2 \Bowtie \mathcal{M}_3$.
For that, we must prove
$L_f(\mathcal{B}_1 \Bowtie \mathcal{B}_3) = L_f(\mathcal{B}_2 \Bowtie \mathcal{B}_3)$.

Let
$w \in L_f(\mathcal{B}_1 \Bowtie \mathcal{B}_3)$
be a finite string on
$\Rec_{\mathcal{N}_1 \cup \mathcal{N}_3}(\mathcal{D})$,
and define the string
$w' = \vis(w \downarrow_{\mathcal{N}_1})$.
By definition,
$Q'_1 \cap F_1 \neq \emptyset$ and $Q'_3 \cap F_3 \neq \emptyset$,
where
$\mathcal{B}_1 \Bowtie \mathcal{B}_3: Q_{01} \times Q_{03} \xrightarrow{w} Q'_1 \times Q'_3$.
Hence, according to lemma~\ref{lem:join-state},
$\mathcal{B}_1: Q_{01} \xrightarrow{w'} Q'_1$
and
$w' \in L_f(\mathcal{B}_1) = L_f(\mathcal{B}_2)$.
Thus $Q'_2 \cap F_2 \neq \emptyset$ where
$\mathcal{B}_2: Q_{02} \xrightarrow{w'} Q'_2$.
Therefore
$w \in L_f(\mathcal{B}_2 \Bowtie \mathcal{B}_3)$.
It follows that
$L_f(\mathcal{B}_1 \Bowtie \mathcal{B}_3) \subseteq L_f(\mathcal{B}_2 \Bowtie \mathcal{B}_3)$.
A similar argument yields
$L_f(\mathcal{B}_2 \Bowtie \mathcal{B}_3) \subseteq L_f(\mathcal{B}_1 \Bowtie \mathcal{B}_3)$.
Thus
$L_f(\mathcal{B}_1 \Bowtie \mathcal{B}_3) = L_f(\mathcal{B}_2 \Bowtie \mathcal{B}_3)$.
\end{proof}

This congruency makes it possible to define the language of finite strings of the join of two LTSRs, regardless of their exact structure, and by considering only their set of traceable finite strings.

\begin{theorem}
Let
$\mathcal{M}_i = (Q_i, \Rec_{\mathcal{N}_i}(\mathcal{D}),\Delta_i, Q_{0i})$
be LTSRs, for $i=1,2$. Then:
\[
\begin{aligned}
L_f(\mathcal{M}_1 \Bowtie \mathcal{M}_2) =
\{ w \in {\textstyle\Rec^*_{\mathcal{N}_1 \cup \mathcal{N}_2}}(\mathcal{D}) |
&\vis(w\downarrow_{\mathcal{N}_1}) \in L_f(\mathcal{M}_1)	\\
\wedge& \vis(w\downarrow_{\mathcal{N}_2}) \in L_f(\mathcal{M}_2) \}.
\end{aligned}
\]
\end{theorem}
\begin{proof}
Let $w \in \Rec^*_{\mathcal{N}_1 \cup \mathcal{N}_2}$ and define
$w' = \vis(w\downarrow_{\mathcal{N}_1})$
and
$w'' = \vis(w\downarrow_{\mathcal{N}_2})$.
By lemma~\ref{lem:join-state},
$\mathcal{M}_1 \Bowtie \mathcal{M}_2: Q_{01} \times Q_{02} \xrightarrow{w} Q'_1 \times Q'_2$
if and only if
$\mathcal{M}_1: Q_{01} \xrightarrow{w'} Q'_1$
and
$\mathcal{M}_2: Q_{02} \xrightarrow{w''} Q'_2$.
Thus $w$ is traceable in $\mathcal{M}_1 \Bowtie \mathcal{M}_2$, if and only if $w'$ is traceable in $\mathcal{M}_1$ and $w''$ is traceable in $\mathcal{M}_2$.
\end{proof}

Let
$\mathcal{M}_i = (Q_i, \Rec_{\mathcal{N}_i}(\mathcal{D}),\Delta_i, Q_{0i})$
be LTSRs, for $i=1,2$. $\mathcal{M}_1$ and $\mathcal{M}_2$ are said to be infinite-traces-based equivalent, denoted by $\mathcal{M}_1 \approx_{it} \mathcal{M}_2$, whenever for any infinite word $w$ on alphabet $\Rec_{\mathcal{N}_1 \cup \mathcal{N}_2}(\mathcal{D})$, $w$ is traceable in $\mathcal{M}_1$ if and only if it is traceable in $\mathcal{M}_2$.
Now we show that the infinite-traces-based equivalence relation on a certain subset of LTSRs is a congruency with respect to the join operation.
A trap state in an automaton is a state with no transition defined from it.
\begin{theorem}
Let
$\mathcal{M}_i = (Q_i, \Rec_{\mathcal{N}_i}(\mathcal{D}),\Delta_i, Q_{0i})$
The equivalence relation $\approx_{it}$ is a congruence relation on set of all labeled transition systems having no trap state, with respect to join operation.
\end{theorem}
\begin{proof}
Suppose
$\mathcal{M}_1 \approx_{it} \mathcal{M}_2$.
It is easy to see that for LTSRs with no trap states, $\approx_{it}$ is a finer equivalence relation than $\approx_{ft}$. Thus for the equivalent Büchi automata,
$\mathcal{B}_i = (Q_i, \Rec_{\mathcal{N}_i}(\mathcal{D}), \Delta_i, Q_{0i}, Q_i)$,
we have $L_B(\mathcal{B}_1) = L_B(\mathcal{B}_2)$ and $L_f(\mathcal{B}_1) = L_f(\mathcal{B}_2)$. Therefore, as a result of lemma~\ref{lem:main-lemma}, for any trapless LTSR $\mathcal{M}_3$,
$L_B(\mathcal{M}_1 \Bowtie \mathcal{M}_3) = L_B(\mathcal{M}_2 \Bowtie \mathcal{M}_3)$.
As stated in proposition~\ref{prop:congruency_binary_commutative}, this is sufficient condition for $\approx_{it}$ to be a congruency with respect to join operation. This completes the proof.
\end{proof}

This congruency makes it possible to define the language of infinite strings of the automaton resulted by the join of two LTSRs, regardless of their exact structure, and by considering only their set of traceable infinite strings.
\begin{theorem}
Let
$\mathcal{M}_i = (Q_i, \Rec_{\mathcal{N}_i}(\mathcal{D}),\Delta_i, Q_{0i})$
be LTSRs, for $i=1,2$. Provided neither $\mathcal{M}_1$ nor $\mathcal{M}_2$ have any trap states:
\[
\begin{aligned}
L_B(\mathcal{M}_1 \Bowtie \mathcal{M}_2) =
\{ w \in {\textstyle\Rec^\omega_{\mathcal{N}_1 \cup \mathcal{N}_2}}(\mathcal{D}) |
&\vis(w\downarrow_{\mathcal{N}_1}) \in L_f(\mathcal{M}_1) \cup L_B(\mathcal{M}_1)	\\
\wedge& \vis(w\downarrow_{\mathcal{N}_2}) \in L_f(\mathcal{M}_2) \cup L_B(\mathcal{M}_2) \}
\end{aligned}
\]
\end{theorem}
\begin{proof}
Let $w \in \Rec^\omega_{\mathcal{N}_1 \cup \mathcal{N}_2}$ and define
$w' = \vis(w\downarrow_{\mathcal{N}_1})$
and
$w'' = \vis(w\downarrow_{\mathcal{N}_2})$.
By lemma~\ref{lem:join-state},
$\mathcal{M}_1 \Bowtie \mathcal{M}_2: Q_{01} \times Q_{02} \xrightarrow{w} Q'_1 \times Q'_2$
if and only if
$\mathcal{M}_1: Q_{01} \xrightarrow{w'} Q'_1$
and
$\mathcal{M}_2: Q_{02} \xrightarrow{w''} Q'_2$.
Thus $w$ is traceable in $\mathcal{M}_1 \Bowtie \mathcal{M}_2$, if and only if $w'$ is traceable in $\mathcal{M}_1$ and $w''$ is traceable in $\mathcal{M}_2$; that is if and only if $w'$ is finite and in $L_f(\mathcal{M}_1)$ or is infinite and in $L_B(\mathcal{M}_1)$ and similarly for $w''$.
\end{proof}

Let
$\mathcal{B}_i = (Q_i, \Rec_{\mathcal{N}_i}(\mathcal{D}),\Delta_i, Q_{0i}, F_i)$
be BARs, for $i=1,2$. We say $\mathcal{B}_1 \approx_{f} \mathcal{B}_2$, whenever $L_f(\mathcal{B}_1) = L_f(\mathcal{B}_2)$.
\begin{theorem}
The equivalence relation $\approx_{f}$ is a congruence relation on set of all BARs with respect to join operation.
\end{theorem}
\begin{proof}
Let
$\mathcal{B}_i = (Q_i, \Rec_{\mathcal{N}_i}(\mathcal{D}),\Delta_i, Q_{0i}, F_i)$
be  BARs for $i=1,2,3$, such that $L_f(\mathcal{B}_1) = L_f(\mathcal{B}_2)$. We show
$L_f(\mathcal{B}_1 \Bowtie \mathcal{B}_3) \subseteq L_f(\mathcal{B}_2 \Bowtie \mathcal{B}_3)$.
Let
$w \in L_f(\mathcal{B}_1 \Bowtie \mathcal{B}_3)$,
and define
$w' = \vis(w \downarrow_{\mathcal{N}_1})$
and
$w''' = \vis(w \downarrow_{\mathcal{N}_3})$.
Lemma~\ref{lem:join-state} states
$\mathcal{B}_1: Q_{01} \xrightarrow{w'} Q'_1$
and
$\mathcal{B}_3: Q_{03} \xrightarrow{w'''} Q'_3$,
whenever
$\mathcal{B}_1 \Bowtie \mathcal{B}_3: Q_{01} \times Q_{03} \xrightarrow{w} Q'_1 \times Q'_3$.
Therefore
$w' \in L_f(\mathcal{B}_1) = L_f(\mathcal{B}_2)$.
Hence
$\mathcal{B}_2: Q_{02} \xrightarrow{w'} Q'_2$,
where
$Q'_2 \cap F_2 \neq \emptyset$. We also know that $Q'_3 \cap F_3 \neq \emptyset$.
Thus, by lemma~\ref{lem:join-state},
$\mathcal{B}_2 \Bowtie \mathcal{B}_3: Q_{02} \times Q_{03} \xrightarrow{w} Q'_2 \times Q'_3$,
where
$Q'_2 \times Q'_3 \cap F_2 \times F_3 \neq \emptyset$.
Therefore
$w \in L_f(\mathcal{B}_2 \Bowtie \mathcal{B}_3)$.
It concludes that
$L_f(\mathcal{B}_1 \Bowtie \mathcal{B}_3) \subseteq L_f(\mathcal{B}_2 \Bowtie \mathcal{B}_3)$.
A similar argument yields
$L_f(\mathcal{B}_2 \Bowtie \mathcal{B}_3) \subseteq L_f(\mathcal{B}_1 \Bowtie \mathcal{B}_3)$.
As a result
$L_f(\mathcal{B}_1 \Bowtie \mathcal{B}_3) = L_f(\mathcal{B}_2 \Bowtie \mathcal{B}_3)$.
By proposition~\ref{prop:congruency_binary_commutative} this completes the proof.
\end{proof}

This congruency makes it possible to define the language of finite strings of join of arbitrary BARs, regardless of their structure, and by considering only their language of finite strings.
\begin{theorem}
Let
$\mathcal{B}_i = (Q_i, \Rec_{\mathcal{N}_i}(\mathcal{D}),\Delta_i, Q_{0i}, F_i)$
be BARs, for $i=1,2$. Then:
\[
\begin{aligned}
L_f(\mathcal{B}_1 \Bowtie \mathcal{B}_2) =
\{ w \in {\textstyle\Rec^*_{\mathcal{N}_1 \cup \mathcal{N}_2}}(\mathcal{D}) |
&\vis(w \downarrow_{\mathcal{N}_1}) \in L_f(\mathcal{M}_1)	\\
\wedge& \vis(w \downarrow_{\mathcal{N}_2}) \in L_f(\mathcal{M}_2)\}.
\end{aligned}
\]
\end{theorem}
\begin{proof}
Let $w \in \Rec^*_{\mathcal{N}_1 \cup \mathcal{N}_2}$ and define
$w' = \vis(w \downarrow_{\mathcal{N}_1})$
and
$w'' = \vis(w \downarrow_{\mathcal{N}_2})$. By lemma~\ref{lem:join-state},
$\mathcal{B}_1 \Bowtie \mathcal{B}_2: Q_{01} \times Q_{02} \xrightarrow{w} Q'_1 \times Q'_2$
if and only if
$\mathcal{B}_1: Q_{01} \xrightarrow{w'} Q'_1$
and
$\mathcal{B}_2: Q_{02} \xrightarrow{w''} Q'_2$.
Therefore $w$ is accepted in $\mathcal{B}_1 \Bowtie \mathcal{B}_2$, if and only if $w'$ is accepted in $\mathcal{B}_1$ and $w''$ is accepted in $\mathcal{B}_2$.
\end{proof}

Finally, we study the problem of congruency of the infinite Büchi-acceptance-based equivalence relation over the set of BARs with respect to the join operation. We demonstrate that infinite Büchi-acceptance-based equivalence does not enjoy congruency.
For sake of concreteness, let
$\mathcal{B}_i = (Q_i, \Rec_{\mathcal{N}_i}(\mathcal{D}),\Delta_i, Q_{0i}, F_i)$
be Büchi automata of records, for $i=1,2$; We say $\mathcal{B}_1 \approx_{B} \mathcal{B}_2$, whenever $L_B(\mathcal{B}_1) = L_B(\mathcal{B}_2)$.

\begin{theorem}
The equivalence relation $\approx_{B}$ is not a congruence relation on set of all Büchi automata of records with respect to join operation.
\end{theorem}
\begin{proof}
According to lemma~\ref{lem:main-lemma}, it is not enough for two Büchi automata to have the same language of infinite strings, but their language of finite strings should be equal as well. In fact, figure~\ref{fig:counterexample-buchi-congruence} gives an example of two trapless Büchi automata, having the same language of infinite strings, but not the same language of finite strings. Thus $\approx_B$ is not a congruency relation over Büchi automata of records with respect to join operation.
\end{proof}

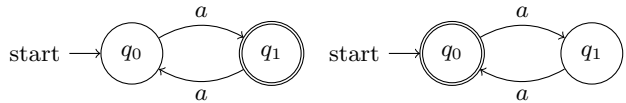
\begin{figure}
\centering
\begin{subfigure}
\centering
\begin{tikzpicture}[auto]
\node[initial,state] (A) {$q_0$};
\node[state,accepting] (B) [right=of A] {$q_1$};
\path (A) edge[->,bend left] node{$a$} (B);
\path (B) edge[->,bend left] node{$a$} (A);
\end{tikzpicture}
\end{subfigure}
\begin{subfigure}
\centering
\begin{tikzpicture}[auto]
\node[initial,state,accepting] (A) {$q_0$};
\node[state] (B) [right=of A] {$q_1$};
\path (A) edge[->,bend left] node{$a$} (B);
\path (B) edge[->,bend left] node{$a$} (A);
\end{tikzpicture}
\end{subfigure}
\caption{Two trapless Büchi automata, having the same language, but not the same language of finite strings}
\label{fig:counterexample-buchi-congruence}
\end{figure}

The above theorem suggests that there can be no join of two BARs cannot be defined as a function of their language of infinte strings.

\section{Conclusions}
The usefulness of an equivalence relation over a set of models in some theoretical-practical purposes such as compositional reasoning, model minimization with the aim of system or code refinement and verification by theorem proving or model checking, requires the equivalence relation to be a congruence relation with respect to the composition operators which are used to construct more complex systems out of the simpler ones~\cite{n16,Izadi-IJCM}. We have shown that using records as a mathematical notion for simultaneous data communications over a set of port names is a powerful way of expressing the external behaviors or coordination parts of computing systems. Thus, labeled transition systems of records are interesting models of coordination systems and we need to consider several equivalence relations over them as their formal semantics. In this paper, we considered four equivalence relations over the set of labeled transition systems of records and investigated their congruency with respect to the join operator. We proved that finite-traces-based, infinite-traces-based, and nondeterministic  automata based equivalence relations are congruence relations with respect to the join operator, while the Büchi automata based relation is not so. Also, using the positive results we introduced the language theoretic definitions for the join operation for the first three cases and showed that there is no such language based definition for the join operation over the set of  Büchi automata of records.

As the future works, we can investigate the problems of being congruence for some other equivalence relations over the set of labeled transition systems of records with respect to the  join or any other composition operator, especially for the relations that are applicable in the process of compositional minimization and model checking based verification such as failure based equivalences like CFFD and NDFD relations~\cite{Valmari-Faster than,Izadi-IJCM} and weak or strong (bi)simulation relations~\cite{milner book}. As another set of future works, defining action (record) based linear and branching time temporal logic systems over label transition systems of records and investigating that fragment of each temporal logic which is preserved using each one of the above mentioned equivalence relations are of our interests.

\end{document}